\documentclass[runningheads,11pt]{llncs}
\usepackage{a4wide,times}
\usepackage{lmodern}
\usepackage{amsmath}
\usepackage{amssymb}
\usepackage{amsfonts}
\usepackage{xfrac}
\usepackage{mathtools}
\usepackage[utf8]{inputenc}
\usepackage{BOONDOX-cal}
\usepackage{enumitem}
\usepackage{thmtools}
\usepackage{graphicx}
\usepackage{cite}

\usepackage[bookmarks=true,hypertexnames=true]{hyperref} 

\usepackage{enumitem}
\newlist{abbrv}{itemize}{1}
\setlist[abbrv,1]{label=,labelwidth=1in,align=parleft,itemsep=0.1\baselineskip,leftmargin=!}

\newcommand{\eps}{\varepsilon}
\newcommand{\opt}{\textsf{opt}}
\newcommand{\copt}{C_{\textsf{opt}}}
\newcommand{\sopt}{s_{\textsf{opt}}}
\newcommand{\pho}{^\rho}
\newcommand{\oeps}{1+\eps}
\newcommand{\prob}{\textsf{ums}}

\begin{document}
	\title{EPTAS for Load Balancing Problem on Parallel Machines with a Non-renewable Resource\thanks{Supported in part by ISF - Israeli Science Foundation grant number 308/18.}}
	\author{G. Jaykrishnan\inst{1} \and
		Asaf Levin\inst{2}}
	\authorrunning{G. Jaykrishnan \and A. Levin}
	%
	\institute{Faculty of Industrial Engineering and Management, The Technion, Haifa, Israel.\email{jaykrishnang@hotmail.com} \and
		Faculty of Industrial Engineering and Management, The Technion, Haifa, Israel.\email{levinas@ie.technion.ac.il}}
	
	\maketitle

\begin{abstract}The problem considered is the non-preemptive scheduling of independent jobs that consume a resource (which is non-renewable and replenished regularly) on parallel uniformly related machines. The input defines the speed of machines, size of jobs, the quantity of resource required by the jobs,  the replenished quantities, and replenishment dates of the resource. Every job can start processing only after the required quantity of the resource is allocated to the job. The objective function is the minimization of the convex combination of the makespan and an objective that is equivalent to the $l_p$-norm of the vector of loads of the machines.  We present an EPTAS for this problem.  Prior to our work only a PTAS was known in this non-renewable resource settings and this PTAS was only for the special case of our problem of makespan minimization on identical machines.
\end{abstract}

\section{Introduction}
The problem \prob\ considered is the non-preemptive scheduling of $n$ independent jobs that consume a single resource (which is non-renewable and replenished regularly) on $m$ uniformly related machines. In the setting of uniformly related machines we let $s_i>0$ be the speed  of machine $i$. Each job $j$ has a size $p_{j}> 0$ associated with it and the processing time of job $j$ on machine $i$ is ${p_{j}}/{s_i}$. The speed of the fastest machine is without loss of generality $1$, else the speed of the machines can be converted to relative speeds with respect to the speed of the fastest machine.  
In our problem, each job $j$ consumes $d_{j}\geq 0$ quantity of resource to start its processing. That is, in order to start processing job $j$, the resource should be allocated to the job at the start of the processing of job $j$.  Therefore, jobs may have to wait if sufficient quantity of the resource is not available and thus idle time is essential for this problem. The resource is supplied at $e$ different time points. The resource replenishment time for time period $k$ is denoted by $u_k$ and the quantity of resource supplied at this time point is $q_k\geq 0$.  Both $u_k$ and $q_k$ (for all $k$) are given as part of the input.  In the given instance of the problem we have $q_k>0$ for all $k$, but we construct additional instances in which it is convenient to allow zero resource supplied in some periods.  

A solution of the problem is feasible if one job is assigned to only one of the machines for the entirety of the continuous processing time of that job and no other job should be assigned to that machine during this time, so there is no overlap of time slots of jobs assigned to a common machine (this is the standard non-preemptive requirement). Also, the necessary quantity of the resource required by the jobs should be supplied, thus  for every point of time $t$, the inequality   
\begin{align*}
	\sum_{j:j\text{ starts strictly before }t}d_j
	\leq 
	\sum_{k:u_k < t} q_k ,
\end{align*}
should be also satisfied for feasibility, so the resource is common for all machines.

The parameters mentioned above define an input instance $I$ for the problem. All these parameters  are rational numbers. 

The objective function considered is the minimization of the convex combination of the makespan and the sum of the $\phi^{th}$ powers of the load of the machines, where $\psi\in[0,1]$ is the coefficient of the convex combination, and $\phi>1$. Given the possibility of idle time, the load of a machine is the time of completion of the last job assigned to the machine. That is, the objective value of a schedule $\sigma$ is defined as
\begin{align*}
	obj(\sigma) = \psi\max_i\Lambda_i +(1-\psi)\sum_{i}\Lambda_i^\phi \ .
\end{align*}
where $\Lambda_i $ is the load of machine $i$ in $\sigma$, and $obj$ denotes the objective function.
The sum of the $\phi^{th}$ power of the completion times of the machines is equivalent to the $l_{\phi}$-norm of the completion times of the machines that is a standard objective in the load balancing literature. Define \emph{norm-cost} of a schedule $\sigma$ as 
\[
\Phi(\sigma) \coloneqq \sum_{i=1}^{m}\Lambda_i^\phi\ .
\] 
The objective considered in this work is a unified generalization of the makespan minimization and the $l_{\phi}$-norm of the loads so we consider the two extreme cases $\psi=0$ and $\psi=1$ as the most interesting cases of this objective.  For $\phi\leq 1$ and $\psi=0$, the optimization problem of finding a minimum cost solution where there is a unique replenishment date at time zero is solved easily by allocating all jobs to one of the fastest machines and thus in this work we consider the problem assuming that $\phi>1$.
Problem \prob\ is to find a job assignment function $\sigma$ which assigns a job $j$ to a machine $i$ and a starting time on the machine that is a feasible solution so that $obj$ is minimized.  Our result is an EPTAS for \prob.

\paragraph{Definitions of approximation algorithms.}
A $\rho$-approximation algorithm for a minimization problem is a polynomial time algorithm that always finds a feasible solution of cost at most $\rho$ times the cost of an optimal solution. A polynomial time approximation scheme (PTAS) for a given problem is a family of approximation algorithms such that the family has a $(1+\varepsilon)$-approximation algorithm for any $\varepsilon>0$. An efficient polynomial time approximation scheme (EPTAS) \cite{Cesati97,Downey99,Flum2006} is a PTAS whose time complexity is upper bounded by the form $f(\frac{1}{\varepsilon}) \cdot poly(n)$ where $f$ is some computable (not necessarily polynomial) function and $poly(n)$ is a polynomial of the length of the (binary) encoding of the input. A fully polynomial time approximation scheme (FPTAS) is defined like an EPTAS, with the added restriction that $f$ must be upper bounded by a polynomial in $\frac 1{\varepsilon}$.    Note that our problem generalizes the standard minimum makespan on identical machines that is known to be strongly NP-hard, and thus our problem does not admit an FPTAS (unless P=NP). Hence, an EPTAS is the fastest scheme that can be established for \prob.

\paragraph{Notation and preliminaries.}
In what follows machines are usually specified by the index $i, i=1,2,\ldots,m$, jobs are specified by the index $j, j=1,2,\ldots,n$, and time periods are specified by the index $k, k=1,2,\ldots,e$. Furthermore, let $J$ be the set of all jobs, $M$ be the set of all machines, and $m(s)$ denote the number of machines with speed $s$.
Without loss of generality, $m \leq n$ because otherwise the $m-n$ slowest machines can be removed from the instance. Let $\eps >0 $ be such that $1/\eps$ is an integer, and our goal is to find an approximation algorithm of approximation ratio $1+\eps$ with the required time complexity bound.  Based on standard scaling (of $\eps$) it suffices to exhibit an algorithm that runs in the required time complexity and returns a feasible solution of cost at most $(1+\eps)^c$ times the optimal cost for an arbitrary constant $c >0$.  Note that if we modify a solution for \prob\ so that the new load of every machine is at most $\oeps$ times its old load then the makespan of the new solution is at most $\oeps$ time the old makespan and the norm-cost of the new solution is at most $(\oeps)^{\phi}$ times the old norm-cost.

\paragraph{Previous studies of scheduling problems with non-renewable resources.}
In \cite{LDS,Herr16,Belkaid12} there are studies motivating the setting of non-renewable resources in scheduling problems like problems arising in steel production or in order picking in a platform with a distribution company. 
Slowiński \cite{SLOWINSKI1984} and Carlier and Kan \cite{CARLIER1982} were among the first to study such problems  and considered minimizing the resource consumed as an objective but the former also added the makespan minimization objective. Grigoriev et al. \cite{Grigoriev2005} considered single machine scheduling with single and multiple resources and established a constant ratio approximation algorithm for makespan minimization and maximum lateness minimization. The settings of non-renewable resources was considered in other studies like \cite{Gyorgyi2014,KIS2015,GYORGYI2015a,Gyorgyi2015b,berczi2020,Gyrgyi2017,gyorgyi2019}.  Last and most relevant to our work, Györgyi \cite{GYORGYI2017604} derived a PTAS for parallel machine scheduling with single non-renewable resource and makespan objective on identical machines.  This PTAS is not an EPTAS, and thus our work both improves the time complexity of the known scheme for this problem to an EPTAS and furthermore it generalizes the machine settings to uniformly related machine and the objective to include also the sum of the $\phi$-powers of machines' completion times that was not considered before in the settings of non-renewable resource.  Furthermore, prior to our work there was no approximation algorithm in the literature for minimizing the makespan on uniform machines with non-renewable resource. 

\paragraph{Earlier approximation schemes for special cases of the problem without non-renewable resources.}
A special case of \prob\ is when there are no resource consuming jobs with parallel machines (this is the situation when the entire quantity of the resource is available to all jobs is released at time $0$). For this machine setting Jansen \cite{Jansen10} showed an EPTAS for the makespan minimization objective improving the seminal PTAS developed by Hochbaum and Shmoys \cite{Hochbaum88}. Epstein and Levin \cite{epstein2014} showed an EPTAS for $\ell_p$-norm minimization problem (see also \cite{EL14} for an alternative approach leading to an EPTAS for this problem), and Kones and Levin \cite{Kones2019} established an EPTAS for the problem of load balancing with the objective that generalizes the one considered here. 

The problem of scheduling on identical machines is a special case of uniformly related machines when speed of all machines are equal. Hochbaum and Shmoys \cite{Hochbaum87} showed an EPTAS (see also \cite[Chapter~9]{hochbaum97}) for the makespan minimization objective for this setting. Alon et al. \cite{Alon89} showed an EPTAS for $\ell_p$-norm minimization. Jansen et al. \cite{JKV16} improved the time complexity of the approximation schemes of the earlier results for makespan minimization on identical machines and uniformly related machines and also for $\ell_p$-norm of the vector of machines loads on identical machines.

Lenstra et al. \cite{Lenstra90} showed that for the more general settings of unrelated machines unless P=NP there is no approximation algorithm with an approximation ratio less than $1.5$ for the makespan minimization objective. They also presented a $2$-approximation algorithm to this classical scheduling problem. 

\paragraph{Outline of the scheme.}
We apply geometric rounding to the parameters of the input (Section \ref{rounding}) followed by a guessing step (Section \ref{guessing}) to guess partial information on the optimal  solution. This guessing step overcomes the impossibility to use the dual approximation method for our objective that is not a bottleneck objective, and it also reflects the special properties of our problem.  It is based on a careful characterization of a subset of solutions containing a near optimal solution (Section \ref{sec:char}).  The so-called configuration mixed-integer linear program, MILP, (Section \ref{confs}) is based on configurations  in addition to the guessed information and we solve the MILP to find an optimal solution of this mathematical program. The optimal solution of the MILP is found in polynomial time using \cite{kannan1983,lenstra1983} since the number of integer variables is a constant.  For every possible value for the guessed information, we solve the MILP, and the least cost solution among all solutions to the MILP is converted to a feasible schedule for the scheduling problem whose cost is approximately the cost of the optimal solution of the MILP (see Section \ref{sec:last}).  

\section{Rounding}\label{rounding}
Our scheme starts by applying standard geometric rounding of the input parameters stated below. We apply rounding of the job sizes, of the machine speeds (and as a result also the processing times), and of the replenishment dates.  The rounding is carried out as follows.
For every job $j$, the size of $j$ is rounded up to the nearest integer power of $(\oeps)$, the speed of every machine $i$ is rounded down to the nearest integer power of $(\oeps)$, and for the replenishment dates we first add $\eps \cdot p_{\min}$ to every replenishment date and then we round up to the nearest integer power of $(\oeps)$ where $p_{\min} =\min_j p_j >0$ is the minimum size of a job in the original instance.
The speed of the fastest machine remains $1$ after the rounding since 1 is an integer power of $\oeps$.

That is,
$$ p'_j=(\oeps)^{\lceil \log_{(\oeps)}p_j\rceil}\  \ \ \  s'_i = (\oeps)^{\lfloor \log_{1+\eps}s_i \rfloor}\  \ \ \mbox{and   } \ \ \  u'_k = (\oeps)^{\lceil \log_{(\oeps)}(u_k + \eps p_{\min} )\rceil} \ .$$

Note that as a result of the rounding the processing time of job $j$ on machine $i$ becomes $p'_{ij}$ and satisfies that 
\begin{equation} p'_{ij}=\frac{p'_j}{s'_i} \in \left[ \frac{p_j}{s_i},(\oeps)^2\frac{p_j}{s_i} \right) .\label{rounded_processing_time} \end{equation}
We let $I'$ be the rounded instance.

Let $\alpha$ be the smallest integer such that $(\oeps)^\alpha \geq u'_1\ ,$ and $\beta$ be the smallest integer such that $(\oeps)^\beta \geq u'_e\ , $ and let $\mu=\beta-\alpha$ be the number of time periods in $I'$ where a period is between two consecutive replenishment dates even if the amount of resource supplied in some of these periods (i.e., at the starting time of a period) are perhaps zero.  Then  
\begin{align*} \mu = \beta - \alpha =&{\lceil \log_{(\oeps)}(u_e + \eps p_{\min} )\rceil}  - {\lceil \log_{(\oeps)}(u_1 + \eps p_{\min} )\rceil}  + 1\\
	\leq& \log_{(\oeps)}(u_e + \eps p_{\min} ) - \log_{(\oeps)} \eps p_{\min} + 2	\leq \log_{(\oeps)} \frac{u_e}{\eps p_{\min}} + 3\ ,
\end{align*}
which is upper bounded by a polynomial in the input encoding length when $\eps$ is fixed. 

Last, we note that if there are different replenishment dates that are identical (in $I'$)  we combine identical replenishment dates into one date by summing up the replenished quantity of the dates that are combined.
There are at most $\mu$ replenishment dates. 

With a slight abuse of notation, let $u'_k$ be the $k^{th}$ replenishment date after the above rounding. The resource supplied at time period $k$, in increasing order of $k$, is given by
$q'_k = \sum_{v:u_v\leq u'_k}q_v-\sum_{v:u'_v < u'_k}q'_v $.

Next, we prove that we can assume without loss of generality that the input to our problem is rounded.
\begin{proposition} \label{rounded_instance_solution}
Any schedule $\sigma$ feasible to instance $I$, of makespan $C_{max}$ and norm-cost $\Phi(\sigma)$, can be used to generate a schedule $\sigma'$ which is feasible to the rounded instance $I'$ with makespan at most $(\oeps)^3 C_{max}$ and $\Phi(\sigma')\leq(\oeps)^{3\phi} \Phi(\sigma)$.  Any schedule $\sigma'$ feasible to the rounded instance $I'$, of makespan $C'_{max}$ and norm-cost $\Phi(\sigma')$, is feasible to the original instance $I$ with makespan at most $C'_{max}$ and norm-cost at most $\Phi(\sigma')$.
\end{proposition}
\begin{proof}
Consider first a schedule $\sigma'$ feasible to the rounded instance $I'$.  We use the same assignment of jobs to machines and we start every job $j$ at the starting time of $j$ in $\sigma'$.  Since machines' speeds are rounded down and job sizes are rounded up, the resulting schedule for instance $I$ is a feasible non-preemptive schedule, and it satisfies the resource requirement constraints as the replenishment dates were rounded up (so in $I$ the resource is available not later than it is available in $I'$).  Since the completion time of every job $j$ in $I$ is not later than its completion time in $I'$, the claim regarding the makespan and norm-cost holds.

Consider the other direction and let $\sigma$ be a feasible schedule to $I$. We define the schedule $\sigma'$ by assigning each job $j$ to the same machine that $\sigma$ is using to process $j$.  We still need to define the starting time of every job.  If $\sigma$ used to complete the processing of job $j$ at time $x$, then in $\sigma'$ we complete processing job $j$ at time $(\oeps)^3 \cdot x$.  Note that the claim regarding the makespan and the norm-cost of the resulting schedule follows immediately.  It suffices to show that this is indeed a feasible schedule. 

First, assume that $j'$ was processed (in $\sigma$) after job $j$ on a common machine $i$. Let $C_j$ and $C_{j'}$ be the completion time in $\sigma$ of $j$ and $j'$, respectively.  Then, $C_{j'} \geq C_j +p_{ij'}$ so the completion time of $j'$ in $\sigma'$ satisfies $(\oeps)^3 \cdot C_{j'} \geq (\oeps)^3 C_j +(\oeps)p_{ij'} \geq (\oeps)^3 C_j +p'_{ij'}$ and therefore the schedule is a non-preemptive schedule.    

Next we verify the availability of resource to process the jobs according to $\sigma'$.  Assume that in $\sigma$ the processing of job $j$ on machine $i$ depends on  the resource supplied at point $u_k$ so the starting time of $j$ was at least $u_k$ and its completion time was at least $C_j \geq u_{k}+p_{ij}$.  In $\sigma'$ we have the new completion time $(\oeps)^3 \cdot C_j \geq \eps  C_j + (\oeps)^2 C_j \geq \eps p_{\min} + (\oeps)^2 u_k +  p'_{ij}$.  Therefore, the new starting time of $j$ on machine $i$ in $\sigma'$ is not smaller than $\eps p_{\min} + (\oeps)^2 u_k $  so indeed the $k$-th supply point after the rounding is not earlier than the starting time of job $j$.  Therefore, the resource that $j$ uses to start its processing is available at this new starting time.
\qed \end{proof}

Using the last proposition, we assume without loss of generality that the original instance is already rounded accordingly so with a slight abuse of notation we let $p_j$, $s_i$, $p_{ij}$, $u_k$ and $q_k$ be the (rounded) size of a job $j$, the (rounded) speed of a machine $i$, the (rounded) processing time of job $j$ on machine $i$, the $k^{th}$ (rounded) replenishment date and resource supplied at $u_k$ respectively and last we assume that $I$ is the (rounded) input instance.   In the next sections we provide an EPTAS for rounded instances of \prob.

\section{Characterization of near-optimal solutions\label{sec:char}}
In the known approximation schemes for the makespan minimization objective (for various machine models) one can apply the standard step of guessing the optimal makespan using binary search  and then discard loads that are sufficiently smaller than this guessed value.  Since our objective is different, we are not able to discard small loads and we are not able to use binary search for finding the approximate makespan. Thus we need another tool to characterize near optimal solution.  In this section we provide the necessary characterization for our scheme.

\begin{lemma}\label{makespan_on_fast}
There is a feasible solution whose objective function value is $(\oeps)^{\phi}$ times the optimal cost satisfying that the makespan is attained on a machine of speed at least $\eps^2$.
\end{lemma}
\begin{proof}
Let $\sigma$ be an optimal solution and among all optimal solutions, $\sigma$ satisfies that the number of machines of speed less than $\eps^2$ having load larger than that of machine $1$ is minimized. If no such machine exists, the claim holds for $\sigma$.  Assume that there are such machines  of speed less than $\eps^2$ having load larger than that of machine $1$ and if the claim does not hold for $\sigma$ then at least one of those machines attains the makespan.

Fix a machine $\mathcal{i}$ that attains the makespan so it has a load larger than the load of machine $1$, where $\mathcal{i}$ is a machine of speed smaller than $\eps^2$. Perform the following transformation to $\sigma$. Transfer the jobs from machine $\mathcal{i}$ to machine $1$ to run after time $\Lambda_1$ in the same order of jobs (as they used to be processed on $\mathcal{i}$) such that the starting time of each job on machine $1$ is at least the starting time of the job on machine $\mathcal{i}$ in $\sigma$. In order to satisfy this condition add idle time where necessary. This transformation is feasible because when the jobs from machine $\mathcal{i}$ are transferred to machine $1$ the time period in which a job is processed on machine $1$ is not before the time period in which it was processed on $\mathcal{i}$. 

Let the resulting schedule be $\sigma'$. Let the load of machine $1$ and $\mathcal{i}$ in $\sigma$ be $\Lambda_1$ and $\Lambda_{\mathcal{i}}$ respectively, and the new load of machine $1$ and $\mathcal{i}$ in $\sigma'$ be $\Lambda'_1$ and $\Lambda'_{\mathcal{i}}=0$ respectively. This transformation can result in two cases. 

First, assume that $\Lambda'_1 \leq \Lambda_{\mathcal{i}} $. Observe that if during the transformation we added idle time to satisfy the condition on the time periods of the transferred jobs, then this case applies. Then, the new makespan is not larger than the makespan of $\sigma$, and furthermore, the difference in norm-cost between the new schedule and old schedule is
${\Lambda'_1}^\phi - \left(\Lambda_1^\phi +\Lambda_{\mathcal{i}}^\phi\right) 
	\leq
	\Lambda_{\mathcal{i}}^\phi - \left(\Lambda_1^\phi +\Lambda_{\mathcal{i}}^\phi\right) 
	\leq 0 $. Thus,
the norm-cost is not increased as a result of the transformation.
Therefore, the cost of $\sigma'$ is not larger than the cost of $\sigma$ but the number of machines with speed smaller than $\eps^2$ with load larger than the load of machine $1$ is decreased. This means that $\sigma$ was not an optimal solution or that we get a contradiction to the choice of $\sigma$ among all optimal solutions. Hence this case is not possible.  

The second case occurs when the finishing time of the last job moved to machine $1$ is later than $\Lambda_{\mathcal{i}}$. That is 
$\Lambda_1 \leq \Lambda_{\mathcal{i}} < \Lambda'_1$.
After the transformation, machine $1$ attains the makespan in $\sigma'$ because the new load of machine $1$, $\Lambda'_1$, is larger than the old load of machine $\mathcal{i}$, $\Lambda_{\mathcal{i}}$, which was the machine that attained the makespan in $\sigma$ and the loads of all other machines remain unchanged. The total increase in load of machine $1$, is $s_{\mathcal{i}}\Lambda_{\mathcal{i}}$ which is at most $\eps^2\Lambda_{\mathcal{i}}$. Thus, the new makespan
$ \Lambda'_1 \leq \Lambda_1+\eps^2\Lambda_{\mathcal{i}} \leq \Lambda_{\mathcal{i}}+\eps^2\Lambda_{\mathcal{i}} \leq (\oeps)\Lambda_{\mathcal{i}}$.  
Moreover, 
\begin{align*}
\Phi(\sigma') &= \sum_{i=1}^{m}{\Lambda'_i}^{\phi} = {\Lambda'_{1}}^{\phi} + {\Lambda'_{\mathcal{i}}}^{\phi} + \sum_{i=2,i\not=\mathcal{i}}^{m}{\Lambda'_i}^{\phi}
\leq ((\oeps)\Lambda_{\mathcal{i}})^{\phi} + \sum_{i=2,i\not=\mathcal{i}}^{m}{\Lambda_i}^{\phi}\leq (\oeps)^{\phi}\Phi(\sigma)\ .
\end{align*}
Thus, the cost of $\sigma'$ is at most $(\oeps)^{\phi}$ times the cost of $\sigma$.
Thus, the schedule $\sigma'$ is a near optimal solution such that the makespan is attained on a machine of speed at least $\eps^2$.
\qed \end{proof}

A machine is said to be active for a time interval consisting of some time periods if the machine is processing a job (or a part of a job) during that time interval and it is idle outside of this time interval. Let $$\tilde{\mu} = \left\lceil\log_{(\oeps)}\frac{1}{\eps} +2\right\rceil$$ be a function of $\eps$, and note that $\tilde{\mu}$ is a constant once $\eps$ is fixed.  Our next goal is to show that we can restrict ourselves to schedules in which every machine is active only for a time interval consisting of a  constant number of consecutive time periods. We prove the next lemma by applying time stretching and moving the new idle time of a machine to be at the beginning of the time horizon. 

\begin{lemma}\label{active_periods}
	The feasible schedule $\sigma$ to the instance $I$ resulting from Lemma \ref{makespan_on_fast} can be converted into another schedule $\tilde{\sigma}$ in which first, every machine is active in a time interval consisting of at most $\tilde{\mu}$ consecutive time periods, second,  $obj(\tilde{\sigma})\leq(\oeps)^\phi obj(\sigma)$, and last, the makespan is attained on a machine of speed at least $\eps^2$.
\end{lemma}
\begin{proof}	
We apply the following process on every machine $i$.  We first change the load of $i$ and then schedule the jobs accordingly.  Assume that the load of $i$ in $\sigma$ is $\Lambda_i$, then the new load of $i$ in $\sigma'$ will be exactly $(\oeps) \Lambda_i$ and we construct the schedule by adding idle time of length $\eps \Lambda_i$ at the beginning of the time horizon, so the active time interval of $i$ will start not earlier than $\frac{\eps \Lambda_i}{\oeps}$ and ends not later than $(\oeps)^2 \Lambda_i$. $\sigma'$ is still feasible  because the starting time of the jobs does not decrease and this ensures that there is enough resource available to start processing the job.  After this additional idle time, the schedule of $i$ will be exactly as in $\sigma$ (while delaying all jobs assigned to $i$ by this idle time).
Thus for any machine $i$, that has jobs assigned to it, $i$ has a new load of exactly $(\oeps)\Lambda_i$. So the claim regarding the cost of the solution as well as the existence of a machine of speed at least $\eps^2$ whose load is the (new) makespan is satisfied (the set of machines attaining the makespan is not modified).  It remains to bound the number of consecutive time periods contained in the active time interval of $i$. Noting that the number of consecutive time periods in this time interval is smaller by $1$ from the number of integer powers of $\oeps$ between $\frac{\eps \Lambda_i}{\oeps}$ and $(\oeps)^2 \Lambda_i$, we conclude that it is at most  $\lceil \log_{\oeps} \frac{(\oeps)^2 \Lambda_i}{\frac{\eps \Lambda_i}{\oeps}} -1\rceil = \left\lceil\log_{(\oeps)}\frac{1}{\eps} +2\right\rceil = \tilde{\mu}$.
\qed \end{proof}

Thus, for the succeeding discussions, it suffices to consider schedules that assign jobs to at most $\tilde{\mu}$ consecutive time periods on every machine and the makespan is attained on a machine of speed at least $\eps^2$. 

\section{Guessing}\label{guessing}

Let $\opt$ be a least cost solution among all the near optimal solutions that satisfy Lemma \ref{active_periods}, and such a feasible solution exists when the input instance has a feasible solution. Let the makespan value  of $\opt$ rounded down to an integer power of $\oeps$ be $\copt$ and the speed of the machine on which the makespan is achieved be $\sopt$. Then the makespan of all near optimal solutions that we would like to consider lies in $[\copt,(\oeps)\copt]$.   In the guessing step that we exhibit now we guess the pair $(\copt,\sopt)$.  

We {\em guess} both pieces of information regarding $\opt$. The term guess is used to mean that we perform the following steps of the algorithm for every possible value of this information and we output the feasible solution of minimal cost among all iterations that results in a feasible solution to \prob. We note that our solution (for a given value of the guess) need not satisfy this guessed information but we will use the existence of a near optimal solution corresponding to this guess.	 In the analysis of the approximation ratio of our scheme it suffices to consider the iteration of this exhaustive search in which we used the correct information regarding $\opt$.
We next verify that the number of possibilities of this guessed information (i.e., the number of iterations of the exhaustive enumeration) is upper  bounded by a polynomial in the input encoding length.

\begin{lemma}
The number of possibilities of the guessed information $(\copt,\sopt)$ is at most $n\left(\log_{(\oeps)}{n}+ 2\right)\cdot\left(1 - \log_{(\oeps)}\eps^2\right)$.
\end{lemma}
\begin{proof}
From Lemma \ref{active_periods} the makespan is achieved on a machine of speed at least $\eps^2$, thus the speed of the machine on which makespan is achieved lies in $[\eps^2, 1]$. The number of possible distinct values for $\sopt$ is at most
$\log_{(\oeps)}1 - \left\lceil\log_{(\oeps)}\eps^2\right\rceil +1
\leq 1 - \log_{(\oeps)}\eps^2 $.
Let $i$ be the machine on which makespan is achieved. Then $\copt$ lies in $[{p_{max}}/{s_i}, n\cdot({p_{max}}/{s_{i}})]$, where $p_{max}$ is the largest rounded size among all the jobs assigned to machine $i$, and $s_{i}$ is rounded speed of machine $i$. Thus the number of possible distinct values for $\copt$ on machine $i$ for a given value of $p_{max}$ and $\sopt = s_i$ is at most
$\left\lceil\log_{(\oeps)}\left(n\cdot\frac{p_{max}}{s_{i}}\right)\right\rceil - \log_{(\oeps)}\left(\frac{p_{max}}{s_{i}}\right) + 1 \leq \log_{(\oeps)}{n}+ 2$,
and thus the total number of possible values we need to check for the pair of values $(\copt,\sopt)$ is at most
$n\left(\log_{(\oeps)}{n}+ 2\right)\cdot\left(1 - \log_{(\oeps)}\eps^2\right)$ as we claimed.
\qed \end{proof}

\section{The mixed-integer linear program}\label{confs}
Our scheme is based on solving a mixed-integer linear program (MILP) for every given value of the guessed information.  The MILP is based on configurations and assignment of jobs to configurations.  We first define some preliminaries definitions, then our notion of configurations, and we conclude this section presenting the MILP.

\paragraph{Huge, Big, and Small jobs.}
We classify jobs into huge, big, and small jobs with respect to the pair $(s,w)$, where $s$ is the speed of a machine and $w$ is its load. The formal definition applies for an arbitrary pair of positive reals. Let $\rho \geq 1$ be a constant that will be specified later.
A job $j$ is a {\em huge job} for a pair $(s,w)$ if ${p_j}/{s} > (\oeps)w$, is a {\em big job} for a pair $(s,w)$ if $(\oeps)w \geq {p_j}/{s} \geq \eps\pho w\ ,$ and is a {\em small job} for a pair $(s,w)$ if ${p_j}/{s} < \eps\pho w$.
Let $\Delta = \left\lceil\log_{(\oeps)} (\oeps)w\right\rceil$ and $\delta =\left\lceil \log_{(\oeps)} \eps\pho w\right\rceil$. Then, the number of distinct sizes for big jobs for a fixed pair $(s,w)$ is at most
$\lambda = \Delta - \delta + 1$,
and $\lambda$ depends only on $\eps$ (for a constant $\rho$).

\paragraph{Large and Small loads.}
Recall that $\copt$ is the guessed approximated value of makespan of a near optimal solution that satisfies Lemma \ref{active_periods}.
A load $w$ of a machine  is a {\em large load} if $w\geq \kappa \copt\ ,$ else it is a {\em small load}, where $\kappa$ is a function of $\eps$ which will be specified later. Thus a large load lies in $[\kappa \copt, (\oeps)\copt]$.  Therefore, given  a value of $\copt$, the number of distinct values of large loads is at most 
$\log_{(\oeps)}\frac{(\oeps)\copt}{\kappa \copt} + 1 = \log_{(\oeps)}\frac{1}{\kappa} + 2$
which is a constant for a fixed value of $\eps$ and $\kappa$.

\paragraph{Fast and slow speeds.}
Let $S$ be the set of all values for the rounded speed of machines.
A speed $s$ or a machine of speed $s$ is {\em fast} if $s\geq\kappa$, else it is {\em slow}, where $\kappa$ is the same function of $\eps$ we used to define large loads. 
The number of distinct values of fast speeds is at most
$\log_{(\oeps)} 1 -\log_{(\oeps)} \kappa + 1 \leq \log_{(\oeps)} \frac{1}{\kappa} + 2$
which is constant when $\eps$ is fixed.

\paragraph{Configuration.}
We use the term {\em configuration} as a compact and approximated representation of a schedule of one machine.  Each configuration $c$ is a vector with $\tilde{\mu}+2$ components (see Figure \ref{conf-fig} for an illustration). 

\begin{figure}[h]
	\centering
	\includegraphics[scale=0.2]{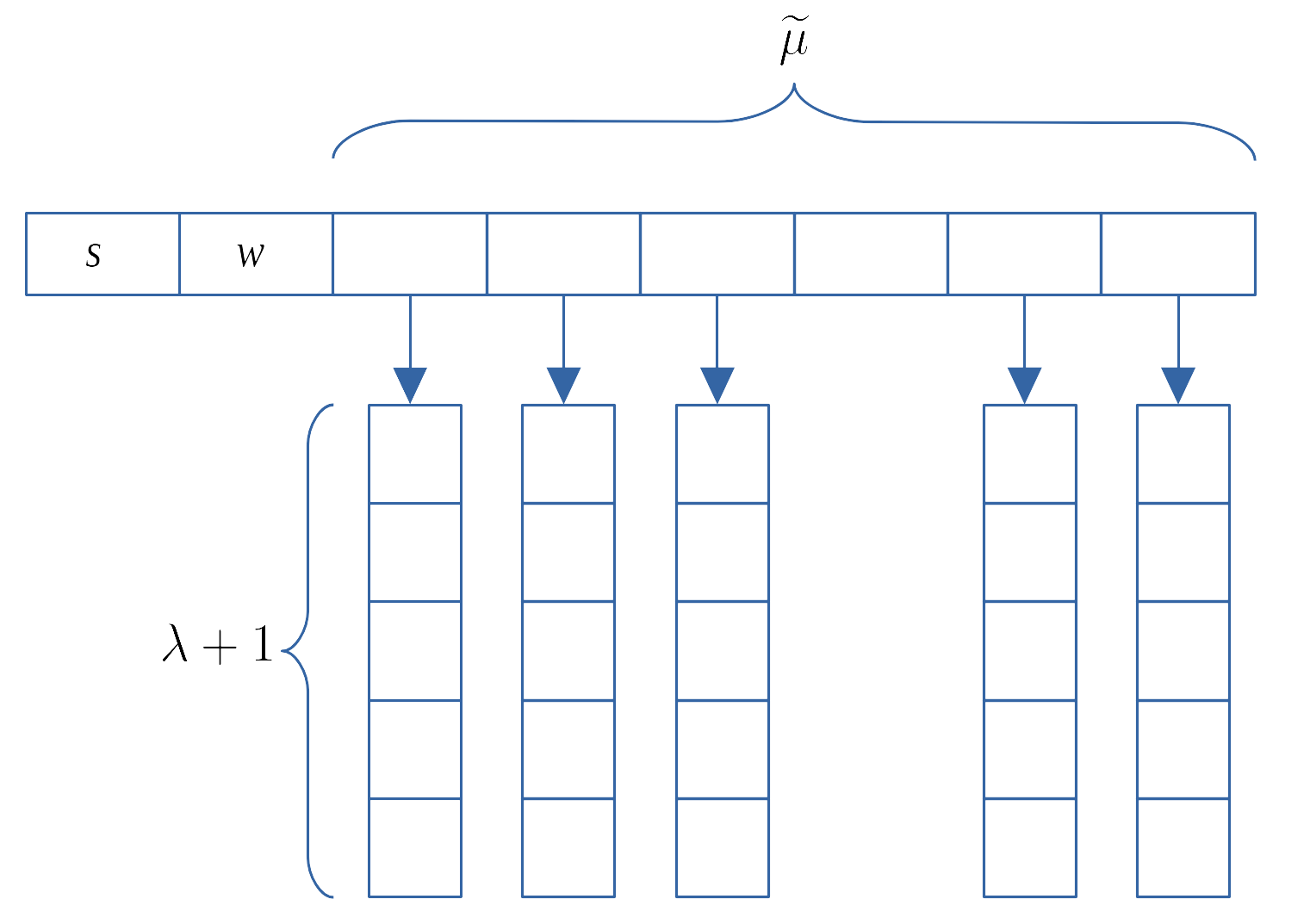}
	\caption{Structure of a configuration\label{conf-fig}}
\end{figure}

The first component stores the speed $s_c$ of the machine assigned configuration $c$. The second component, $w_c$, is the upper bound on the load (load bound) of a machine to which configuration $c$ is assigned, rounded up to an integer power of $(\oeps)$. Based on our characterization of near optimal solutions, a machine assigned $c$ will be active for at most $\tilde{\mu}$ time periods intersecting $[\eps w_c, (\oeps)w_c]$. 
Let $\Delta_c = \left\lceil\log_{(\oeps)} (\oeps)w_c\right\rceil$ and $\delta_c =\left\lceil \log_{(\oeps)} \eps\pho w_c\right\rceil$. The number of distinct sizes for big jobs for the configuration $c$ is at most
$\lambda_c = \Delta_c - \delta_c + 1$.

Let $\nu_c^l = \left\lfloor\log_{(\oeps)}\eps w_c\right\rfloor$ and let $\nu_c^u = \left\lceil\log_{(\oeps)}(\oeps) w_c\right\rceil$. Each of the remaining $\tilde{\mu}$ components of the configuration corresponds to time periods $k = \nu_c^l,\nu_c^l+1,\nu_c^l+2,\ldots,\nu_c^u$. 
For each period $k$ in this interval, there is an {\em associated subvector} of $\lambda_c+1$ components (a column in Figure \ref{conf-fig}). Each of the first $\lambda_c$ components of the subvector associated to a period $k$ stores the number of big jobs, for the pair $(s_c,w_c)$, of processing time $(\oeps)^{\delta_c+r-1}$ on machines with speed $s_c$ for $r=1,2,\ldots,\lambda_c$ that are started in that period on a machine assigned configuration $c$, and we denote this value of the component by $B_c(k,r)$. The last component of the subvector associated to a period $k$ stores the floor of the ratio of the sum of the processing time of all the small jobs, for the pair $(s_c,w_c)$, that are started in that period on a machine assigned configuration $c$ and $\eps\pho w_c$, denoted by $S_c(k)$. That is,
\[
S_c(k) = \left\lfloor\frac{\sum ({p_j}/{s_c})}{\eps\pho w_c}\right\rfloor\ ,
\]
where the summation is over the set of the small jobs assigned to period $k$ and configuration $c$ on such machine.

\paragraph{Number of configurations.}
Our next goal is to upper bound the number of configurations with a common given pair $(s_c,w_c)$ by a constant depending on $\eps$ and to upper bound the number of configurations for all such pairs by a polynomial in the input encoding length.  We will prove the following upper bounds.

\begin{lemma}
For a fixed value for speed and load bound, the number of possible configurations is at most 
\[
	\left(\frac{1+\eps}{\eps\pho}+1\right)^{\tilde{\mu}(\lambda+1)}\ .
\]
In total there are polynomial number of configurations.
\end{lemma}
\begin{proof}
The number of big jobs that start in a time period is at most the number of big jobs that can be assigned to a machine, which is at most
$\frac{(\oeps)w_c}{\eps\pho w_c} = \frac{1+\eps}{\eps\pho}$.
Thus, the number of distinct values that $B_c(k,r)$ can take is at most 
$\frac{1+\eps}{\eps\pho} + 1$.
The number of distinct values that $S_c(k)$ can take is at most 
$\frac{(\oeps)w_c}{\eps\pho w_c} + 1 = \frac{1+\eps}{\eps\pho} + 1$.
For a fixed value for speed and load bound, the number of possible configurations is at most 
$\left(\frac{1+\eps}{\eps\pho}+1\right)^{\tilde{\mu}(\lambda+1)}$ as we claimed. 

The number of distinct values that the load of a configuration given a fixed speed of the configuration can take can be upper bounded by the following argument.  There are $n$ options for the maximum processing time of a job assigned to this configuration, and the load bound is between this processing time and $n$ times this processing times. The number of distinct values of speed of a configuration is at most the number of machines. Thus the total number of configurations is at most 
$ \left(n\left(\log_{(\oeps)}{n}+ 2\right)\right) \cdot m \cdot\left(\frac{1+\eps}{\eps\pho}+1\right)^{\tilde{\mu}(\lambda+1)}$,
which is upper bounded by a polynomial in the input encoding length when $\eps$ is fixed.
\qed \end{proof}

\paragraph{Generating a schedule of one machine from a configuration $c$.}  We would like to restrict ourselves to configurations that represent a non-preemptive schedule.  To do that, we define a partial schedule of one machine corresponding to a configuration $c$, and the possibility to create this partial schedule will mean that the configuration will be called {\em feasible}.
Based on the configuration $c$ with its speed $s_c$ and the load bound $w_c$, we create a schedule on $\tilde{\mu}$ periods intersecting $[\eps w_c, (\oeps)w_c]$ on that machine.
For every period $k$ in this interval, in an increasing order of $k$, perform the steps below.
\begin{itemize}
\item Initialize the starting time of the first job in period $k$ denoted as $Start_k$ to be 
$$\max\{u_k,\mbox{ the finishing time of the last processed job in period at most } k-1\} .$$ When $k$ is the first period in this interval, finishing time of last processed job in period $k-1$ is the starting time of period $k$.
	\item If $S_c(k)$ is  not zero, then assign a {\em virtual job} of size $S_c(k)\cdot \eps\pho w_c s_c$ to start at time $Start_k$ and increase the  value of $Start_k$ by the processing time of this virtual job. 
		\item The remaining jobs assigned to period $k$ are big jobs and these big jobs are scheduled sequentially without idle time in non-decreasing order of their size, the number of big jobs of each size are as described in $c$, and the first such job is scheduled to start at $Start_k$. 
\end{itemize}

A configuration is {\em feasible} if  for each period in the schedule generated, there is at most one job (that is not a virtual job) whose processing crosses into the successive period(s).
Let $C$ be the set of all feasible configurations and we compute $C$ in advance with the required time complexity of the form $f(\frac{1}{\eps})\cdot poly(n)$.  We can indeed compute $C$ in advance in this required time complexity as the algorithm for generating a schedule of one machine based on a configuration is a linear time algorithm and verifying feasibility is carried out during this generation process.

\paragraph{Notation for the MILP.}

A MILP is developed that uses the guessed information of $\opt$ to schedule the jobs to configurations for the rounded instance, $I'$. For each configuration $c \in C$, the MILP assigns the jobs to only a particular set of consecutive periods such that when that configuration is assigned to a machine, the machine will be active for only that particular set of consecutive periods. The set of periods is defined by the configuration. Let
\[
	\nu_c^l = \lfloor\log_{(\oeps)}\eps w_c\rfloor\ ,\ 
	\nu_c^u = \lceil\log_{(\oeps)}(\oeps) w_c\rceil\ .
\] 
Then the set of consecutive time periods for which a machine assigned configuration $c$ will be active (or more precisely it will not be active in other time periods) is denoted by $K_c = \{(\oeps)^k:k = \nu_c^l,\nu_c^l+1,\nu_c^l+2,\ldots,\nu_c^u\}\ ,\forall c\in C$. 
Furthermore, let
 $K=\{(\oeps)^i: i = \log_{\oeps} u_1, \log_{\oeps} u_1+1, \ldots, \log_{\oeps} \copt +1\},$ denotes the set of universal time periods for any configuration and recall that $|K| = \mu$ is bounded by a polynomial in the input encoding length. 

Since the classification of the jobs depend on both the speed of the configuration and the load of the configuration the distinct size of the big jobs change with respect to the configuration. Hence let $R_c\ ,\forall c\in C$ denote the set of indexes for integer power of $\oeps$ corresponding to sizes of big jobs for the pair $(s_c,w_c)$ of configuration $c$.  

For each speed $s$ and load bound $w$ that is an integer power of $\oeps$, create the sets $HJ(s,w)$, $BJ(s,w)$ and $SJ(s,w)$ for huge jobs, big jobs and small jobs for the pair $(s,w)$ respectively. Order the jobs in each list in non-decreasing order of the resource requirements. 

\paragraph{The decision variables of the MILP.}  We have two families of decision variables.

{\bf Configuration counters.} Denoted by $z_{c}$. The variable $z_c$ counts the number of machines assigned configuration $c\in C$.
Each configuration counter for a configuration of fast speed and large load is required to be integer and  configuration counters for remaining configurations are allowed to be fractional.
Thus the number of integer configuration counters is at most 
$\left(\log_{(\oeps)} \frac{1}{\kappa} + 2\right)^2\cdot\left(\frac{1+\eps}{\eps\pho}+1\right)^{\tilde{\mu}(\lambda+1)}$,
which is a constant when $\eps$ is fixed. 
The number of fractional configuration counters is polynomial in the input encoding length.

{\bf Assignment variables.} Denoted by $x_{jkc}$. In an integer solution, if $x_{jkc}$ is 1, job $j$ is assigned to configuration $c$ and to start in period $k$. All these variables are allowed to be fractional. 
The number of assignment variables is at most
$n\cdot \vert K \vert \cdot \vert C\vert $,
which is upper bounded by a polynomial in the input encoding length.

\paragraph{The MILP.} 
The objective function of the MILP is the minimization of
\begin{align*}
\psi\copt +(1-\psi)\sum_{c\in C}z_c\cdot w_c^\phi\ .
\end{align*} 

We next list the constraints of the MILP (together with their meaning).

 A job can be assigned to a configuration and period combination only as a non-huge job.
 	\begin{align}
 	\sum_{c\in C:j\not\in HJ(s_c,w_c)}\sum_{k\in K_c} x_{jkc} = 1\ ,\forall j\in J \ .\label{only_one_assignment} 
 	\end{align}
 	
  For each size, the number of big jobs of that size assigned to a period $k$ of a configuration $c$ is equal to the number of big jobs of that size in period $k$ of  configuration $c$ times the number of machines of this configuration.
 	\begin{align}
 	&\sum_{j:j\in BJ(s_c,w_c),p_{j}=(\oeps)^{\delta_c+r-1}} x_{jkc} = z_c\cdot B_c(k,r)\ ,\forall c\in C, k\in K_c, r\in R_c\ . \label{big_jobs}
 	\end{align}
 An upper bound on the sum of processing time of small jobs (for every configuration and every period), while allowing some additional space that is needed since $S_c(k)$ was obtained by rounding down.
 	\begin{align}
 	&\sum_{j\in SJ(s_c,w_c)} (x_{jkc}\cdot p_j) \leq z_c\cdot s_c\cdot (S_c(k)+1)\cdot\eps\pho w_c\ ,\forall c\in C,k\in K_c\ .  \label{small_jobs}
 	\end{align}
 The number of chosen configurations for a particular speed should be equal to the number of machines of that speed.
 	\begin{align}
 	\sum_{c\in C:s_c=s} z_c &= m(s) \ ,\forall s \in S\ . \label{configurations_machines}
 	\end{align} 
In each period, the resource requirement of a job is met before the job starts processing. Resource can be carried over to future periods, if available.
 	\begin{align}
 	\sum_{c\in C}\sum_{k':k'\leq k}\left(\sum_{j}(x_{jk'c}\cdot d_j)\right)\leq \sum_{k':k'\leq k}q_{k'}\ ,\forall k\in K\ . \label{resource_constraint}
 	\end{align}
Constraint that enforces the guessing of the makespan and the machine on which makespan is achieved, once again we slightly relax the condition and allow a machine of that speed with slightly smaller load bound:
 \begin{align}
 \sum_{c\in C:s_c=\sopt,(\oeps)w_c\geq\copt} z_c\geq 1\ .\label{guessing_constraint}
 \end{align}
The bounds for the assignment variable.
 	\begin{align}
 	0 \leq x_{jkc}\leq 1&\ ,\forall c\in C, k\in K, j\in J\label{0_x_1}\ .
 	\end{align}
For configurations with fast speed and large load, the configuration counter variable is an integer and all configuration counters are non-negative.
 \begin{align}
 	&z_c \geq 0\ ,\forall c\in C\ ,\text{ and}\label{nonnegative_z}\\
 	&z_c \in \mathbb{Z} \ ,\forall c\in C: s_c\geq \kappa, w_c\geq \kappa\copt\ .\label{integer_z}
 \end{align}

\section{Using the MILP\label{sec:last}}

The remaining parts of our scheme and its analysis are as follows.  We first show in Theorem \ref{theorem_1} that (for the correct guessed value) there is a solution to the MILP whose objective function value is at most $(\oeps)^{\phi}$ times the objective value of the near optimal solution satisfying Lemma \ref{active_periods}. That is, we prove that the MILP has a feasible solution of cost not significantly larger than the cost of the near optimal schedule of our problem.  Then, our scheme continues by finding an optimal solution for the MILP, and then transform it into a feasible schedule (a feasible solution to our scheduling problem) of cost not much larger than the cost of the given optimal solution for the MILP.  The use of the solution of the MILP for constructing our output is considered in subsection \ref{sec:thm2}.

\begin{theorem}\label{theorem_1}
The optimal objective function value of the MILP is at most $(\oeps)^{\phi}$ times the objective value of the near optimal solution satisfying Lemma \ref{active_periods} (as a solution to the scheduling problem). 
\end{theorem}
\begin{proof}
Assume that there exists a near optimal schedule $\sigma$ to instance $I'$ that satisfies Lemma \ref{active_periods}.  
Create a configuration for each machine $i$ and with a slight abuse of notation we let $i$ denote also the configuration for this machine. We create this configuration as follows. For the first component of the configuration $i$ store the speed of the machine, $s_i$, for the second component store the load of the machine rounded up to an integer power of $(\oeps)$, $w_i$, and create $\tilde{\mu}$ components. For each period $k$ in the interval containing all active periods, identify the number of big jobs for each distinct size of big jobs for $(s_i, w_i)$ assigned to that machine in period $k$, and save that as $B_i(k,r)$ (in configuration $i$). For each period $k$ find 
\[
	S_i(k) = \left\lfloor\frac{\sfrac{\sum p_j}{s_i}}{\eps\pho w_i}\right\rfloor\ ,
\]
where the summation is over the small jobs for $(s_i,w_i)$ assigned to machine $i$ in period $k$. 

Thus,  this defines  a configuration for each machine $i$. Two configurations are identical if each component of the two configurations are equal. Let $C'$ be the set of distinct configurations. 
In order to check the feasibility of the configurations in $C'$, create a schedule from a configuration $c\in C'$. Assign the configuration to a machine $i$ of speed $s_c$. The configuration has $\tilde{\mu}$ active time periods and let the active periods start at time period $k_{\ell}$. First, we assign idle time before period $k_{\ell}$ to machine $i$. Then, for each time period $k\in [k_{\ell},k_{\ell}+\tilde{\mu}-1]$ perform the following in an increasing order of $k$. Assign {\em sand} representing small jobs assigned to the configuration $c$ to machine $i$. The total size of the sand in period $k$ is at most $S_c(k)\eps\pho \cdot w_c\cdot  s_c$. Then, assign the big jobs in a non-decreasing order of size, $(\oeps)^{\delta_c+r-1}$. Hence, in a non-decreasing order of $r$ we assign $B_c(k,r)$ big jobs of size $(\oeps)^{\delta_c+r-1}$ sequentially after the processing of the sand. The above assignment of jobs to periods is feasible due to the following reasons.  First, the schedule $\sigma$ is feasible; second, the size of the small jobs was rounded down while the creation of the schedule so the total size of small jobs cannot increase, with respect to the total size of small jobs assigned by $\sigma$; and last the jobs are assigned sequentially in non-decreasing order of their size.  Therefore, the only job that might cross over into the succeeding period is a job of the largest size and it must be a large job.  Thus, $C'\subseteq C$.
Define $z_c\in\mathbb{Z}$, for every configuration $c\in C$, by letting $z_c$ be the number of machines assigned configuration $c$ (this is zero if $c\in C\setminus C'$).
Define $x_{jkc}\in\mathbb{Z}$ for $j\in J, k \in K, c\in C$, where $K$ is the set of all possible time periods as follows. For each job $j$, $x_{jkc} = 1$ if $j$ is assigned to a machine with configuration $c$ in period $k$, else $x_{jkc}=0$.  Next, we establish that this is a feasible solution for the MILP.

Constraint \eqref{only_one_assignment} and \eqref{0_x_1} are satisfied because each job is assigned to only one machine and only once by definition of $x_{jkc}$. Constraint \eqref{big_jobs} is satisfied by the definition of $x_{jkc}$, $z_c$, and $B_c(k,r)$.  In the definition of $S_c(k)$ the summation is over the set of small jobs with respect to $(s_c,w_c)$ assigned to start in a period. Thus for each configuration $c$
$\frac{\sum\frac{p_j}{s_c}}{\eps\pho w_c}-1 \leq S_c(k)$ so
$\sum p_j \leq s_c(S_c(k)+1)\eps\pho w_c$, and taking into account identical configurations, we get $\sum p_j \leq z_c\cdot s_c\cdot (S_c(k)+1)\cdot\eps\pho w_c$,
where the summation in the last inequality is over small jobs for $(s_c,w_c)$. Thus, Constraint \eqref{small_jobs} is satisfied.  Constraint \eqref{configurations_machines} is satisfied by definition of $z_c$.  Constraint \eqref{resource_constraint} is satisfied because $\sigma$ is a feasible schedule. The configuration of the machine on which the makespan is achieved ensures that constraint \eqref{guessing_constraint} is satisfied since $\sigma$ satisfied Lemma \ref{active_periods}. Constraints \eqref{nonnegative_z} and \eqref{integer_z} are satisfied by definition of $z_c$ as all these variables are non-negative integers. 

Thus the configurations $C$ are all feasible configurations, and the solution we have defined $(\textbf{x},\textbf{z})$ is a feasible solution to the MILP. The objective value (of the MILP) for the constructed solution is at most
\begin{align*}
\psi (\oeps) \max_{c\in C'}w_c + (1-\psi)\sum_{c\in C}z_cw_c^\phi\ ,
\end{align*} 
since there is a configuration in $C'$ with load bound of at least $\frac{\copt}{\oeps}$.  
Note that for every machine $i$ whose configuration is $i$, we have that the load bound in the configuration $w_i$ is at most $\oeps$ times the load of $i$ in $\sigma$, so the claim regarding the cost of the solution to the MILP in terms of the cost of $\sigma$ holds.
\qed \end{proof}

\subsection{Transforming the MILP solution into the output \label{sec:thm2}}
The transformation of an optimal solution to the MILP to a near optimal solution for \prob\ that we describe in the rest of this section proves the following theorem.

\begin{theorem}\label{theorem_2}
If there exists a solution to the MILP for  $I$ for the given value of the guessed information, then there exists a feasible schedule to instance $I$ of cost at most $(1+6\eps)^{\phi}$ times the cost of the solution of the MILP. 
\end{theorem}

Let the MILP solution be $\textsf{sol}$,  the values of the decision variables in this MILP solution be $(\textbf{x},\textbf{z})$, and the guessed information be $\copt$ and $\sopt$. $\textbf{x}\in\mathbb{R}$ but the $z$ value is integer for configurations with fast speed and large load. Let $\hat{z}_{c} = \left\lceil z_{c}\right\rceil\ ,\forall c\in C$. This converts all fractional configuration counters to integers while the configuration counters for fast speed and large loads remain unchanged and integer. Thus for every speed $s$ we add $\sum_{c\in C:s_c=s}\hat{z}_c - m(s)$
{\em virtual machines}. Virtual and real machines in this discussion do not refer to the machines in the instance $I$, and only act as terms to explain the last rounding phase of our scheme. Thus, now we have more machines than specified in the input instance. 

Next, we assign configurations to machines, i.e., every machine $i$ is now associated with a configuration $c$ whose first component is the speed of $i$ such that the number of machines that are assigned a configuration $c$ is equal $\hat{z}_c$. We allocate the configurations to real and virtual machines in a way that for every configuration $c$, at least $\lfloor z_c \rfloor$ machines that were assigned configuration $c$ are real machines (and at most  one machine assigned configuration $c$ is a virtual machine).  This last requirement is possible as the $z$ values satisfy constraint \eqref{configurations_machines}. For further discussions let $C$ be the set of configurations $c$ with $\hat{z}_c>0$. 

Fix a real machine of speed $\sopt$ as the machine which attains the makespan by assigning the configuration of largest load of speed $\sopt$ to this machine. From constraint \eqref{guessing_constraint}, the load bound of this machine is at least $\copt/(\oeps)$. Let that machine be denoted as $\mathcal{i}$.

We first show how to allocate the jobs to the machines (virtual or real) such that the load of every such machine is at most $1+O(\eps)$ times its load bound and so that the resource required by the jobs is supplied on time.  Later on we will analyze the impact of not using virtual machines.

\paragraph{Allocating most jobs and leaving only some jobs that are small for at least one machine.}
Recall that  $K$ denotes the set of universal time periods for any configuration.
Let $R'$ denote the set of integer powers of $(\oeps)$ corresponding to distinct sizes of all jobs in the (rounded) instance $I$. Let $\alpha_{r'k}$ be the number of jobs of size $(\oeps)^{r'}, r'\in R'$, assigned to start in period $k$ based on the MILP solution, that is
\begin{align*}
\alpha_{r'k} = \left\lfloor\sum_{j\in J:p_j=(\oeps)^{r'}}\  \sum_{c\in C}x_{jkc}\right\rfloor\ .
\end{align*}
Thus according to the MILP solution there is enough resource to start processing $\alpha_{r'k}$ jobs of size $(\oeps)^{r'}$ in period $k$ and furthermore there is enough available time to process so many jobs of each size on the machines of $I$. Intuitively, since $\alpha_{r'k}$ is rounded down, for every size for which there is a machine with speed and load bound for which the size is a size of a small job, there is at most one job of this (small) size left unassigned (with respect to the MILP solution) in each such period.

For every machine and every active period, in the first rounding step, we increase the total size available for small jobs in period $k$ by $2\eps\pho s_cw_c$, where $c$ is the configuration assigned to the machine. This is done to ensure that the process we present will be feasible. Now the total size available for small jobs on a machine in period $k$ is $(S_c(k)+2)\eps\pho s_cw_c$, where $k$ is one of the active periods of that machine.   We will allow adding the additional small jobs at the end of the schedule of this machine. 

Next in the second rounding step that we describe below, we assign almost all jobs to the real or virtual machines where the load of every machine with configuration $c$ will be approximately $w_c$.  More precisely, for each period $k\in K$, in an increasing order of $k$, perform the following.

Create a list of jobs $J_k$ for which there exists $c\in C$ such that $x_{jkc}>0$ and $j$ was not scheduled to start in an earlier period. Sort $J_k$ in non-decreasing order of the resource requirement.
The assignment of big jobs is detailed next.
Go over each machine one by one and perform the following operations. Identify the set of sizes of jobs that are big for the machine assigned configuration $c$ and let $R_c$ denote the set of indexes for integer powers of $\oeps$ corresponding to sizes of big jobs for the pair $(s_c,w_c)$ for configuration $c$ where $c$ is the configuration of the current machine. For each $r\in R_c$, in an increasing order of $r$, perform the following operations. Identify the jobs of size $(\oeps)^{\delta_c+r-1}$ in $J_k$, in the order the jobs appear in $J_k$. Assign the jobs one by one in order starting with the first job as follows. A big job of size $(\oeps)^{\delta_c+r-1}$ can be assigned to the current machine, assigned configuration $c$, if after assigning the job to this machine and period the following two conditions hold.  First, the total number of jobs of size $(\oeps)^{\delta_c+r-1}$ assigned to all the machines in period $k$ is at most $\alpha_{(\delta_c+r-1)k}$ and second the total number of big jobs of size $(\oeps)^{\delta_c+r-1}$ assigned to the current machine in period $k$ is at most $B_c(k,r)$. If a job is assigned to a machine, remove it from $J_k$.

Once the above procedure is completed for the current value of $k$ we move on to schedule jobs as small jobs (in this period) as follows.  
Go over the machines in non-decreasing order of the product of speed and load bound (of its assigned configuration) and perform the following operations for assigning jobs to the current machine. Identify the set of sizes of jobs that are small for the machine assigned configuration $c$ and let $R'_c$ denote the set of indexes for integer powers of $\oeps$ corresponding to sizes of small jobs for the pair $(s_c,w_c)$ for configuration $c$. For each $r\in R'_c$, in an increasing order of $r$, perform the following operations. Identify the jobs of size $(\oeps)^{\delta_c+r-1}$ in $J_k$, in the order the jobs appear in $J_k$. Assign the jobs one by one in order starting with the first job as follows. Assign the job to the machine if the total size of small jobs on that machine in period $k$, after assignment of the job, is at most $(S_c(k)+2)\eps\pho s_cw_c$, and the total number of jobs of size $(\oeps)^{\delta_c+r-1}$ assigned in period $k$ on all machines, after assignment of the job, is at most $\alpha_{(\delta_c+r-1)k}$.  Once we have considered all machines (for this period $k$), we increase $k$ by one and go back to creating the lists $J_k$.

After completing the assignment of jobs using the above two rules for all $k$, the second rounding step is completed.  
The remaining unassigned jobs are small jobs that were left unassigned due to the rounding down of $\alpha_{r'k}$.  In the last rounding step all non-assigned jobs are assigned to the end of the schedule on machine $\mathcal{i}$ and if necessary we add idle time before processing these jobs to ensure that they do not start before time $\copt$. This assignment of jobs in the third rounding step to the end of the schedule on $\mathcal{i}$ is feasible because there exists enough resource to assign all these jobs by the time machine $\mathcal{i}$ finishes processing its scheduled jobs using the fact that $\copt$  is not smaller than the last time in which a resource is supplied (considering only the resource needed by the set of jobs in the instance). 

In the last rounding step of our scheme all jobs that were assigned to virtual machines in the second rounding step are moved to be processed on machine $\mathcal{i}$ in an arbitrary order without idle time after the last job assigned to that machine (in the second or third step) is completed (once again starting not earlier than $\copt$).  Similarly to the argument regarding the feasibility of the assignment of jobs in the third rounding step, the resulting schedule is feasible if we are able to show that the assignment of jobs in the second rounding step is feasible with respect to the resource constraint.  This completes the description of the procedure of transforming the MILP solution into a feasible schedule for \prob.  We next turn our attention to the analysis of this transformation procedure.

\paragraph{Proving that the partial assignment of jobs to virtual and real machines is feasible.}
The feasibility of the partial schedule at the end of the second rounding step with respect to the resource constraint follows from the feasibility of the MILP solution and noting that in every prefix of periods (i.e., all periods $k'$ with $k'\leq k$ for a fixed value of $k$) and every size of jobs $(\oeps)^r$ (for a fixed $r\in R$) we have the following. The partial schedule schedules no more than $\sum_{k':k'\leq k} \alpha_{rk'}$ jobs of size $(\oeps)^r$ to start in this prefix of periods and these jobs have the least resource requirement (among the jobs of that size). However, by the definition of the $\alpha$ values, the MILP solution schedules at least so many jobs of this size to start during this prefix, and the Constraint \eqref{resource_constraint} guarantees that the total consumption of resource in this partial schedule in this prefix of period does not exceed the total amount of resource that is supplied in this prefix of periods.

Next, we would like to argue that at the end of the second rounding step for every period $k$ we have the following.  For every size of jobs $x$ for which there is no machine active in this period where $x$ is a size of a small job for that machine, the number of jobs of this size assigned to start in this period is exactly as in the MILP solution, and for every other size we are left with at most one unassigned job.  Thus we will prove the following lemma.

We will let $\rho=10$ and require the property that $\kappa \leq\eps^{2\rho+3}$ so if a job cannot be assigned as a small job to machine $\mathcal{i}$ then it must be assigned as a large job to a fast machine with large load and all these jobs are assigned (integrally) by the MILP solution.

\begin{lemma}
For every period $k$, let $S(k)$ be the set of sizes of jobs satisfying that for every $x\in S(k)$, there is at least one configuration $c$ where a job of size $x$ is not huge for $(s_c,w_c)$ and $x$ is small for machine $\mathcal{i}$.  At the end of the second rounding step we have that the set of unassigned jobs have at most one job of each size for every period where this size belongs to $S(k)$ (and no such job of this size if the size does not belong to $S(k)$).
\end{lemma} 
\begin{proof}
Fix a size of jobs $(\oeps)^r$ (for $r\in R$),  we want to show that the number of unassigned jobs of this size is one or zero.
 Consider a configuration $c\in C$ and a period $k \in K$. Notice that $B_c(k,r)\hat{z}_c$ is integer and it upper bounds the number of jobs of this size assigned to a machine of this configuration and period as large jobs in the MILP solution. We argue that the second rounding step assigns exactly $\alpha_{kr}$ jobs of this size to start in period $k$ and this number is exactly $\sum_{j\in J:p_j=(\oeps)^{r}}\  \sum_{c\in C}x_{jkc}$ if $(\oeps)^r \notin S(k)$.  First, consider the second part, namely if $r$ is large enough so that there is no machine in which a job of size $(\oeps)^r$ is a non-huge job, then such jobs cannot be assigned to start at this period in the MILP solution.  
  It suffices to show that we indeed assign $\alpha_{kr}$ jobs of this size for every period $k$ and we would like to prove that for every $r\in R$.
 
Assume that this is not the case.  Since the assignment of jobs have over-used the positions as large jobs, the assignment of jobs as small jobs does not guarantee the feasibility of this assignment for some of the sizes. 
Notice that the machines are ordered in the non-decreasing order of the product of speed and load bound for the assignment of the small jobs. Thus when a job is small for a machine then the job will be small for all succeeding machines in the ordering. 

Identify the first machine $i'$, in the ordering, such that each machine in the suffix of machines starting from $i'$ and up to the end of the (ordered list of) machines, satisfies that the total size of small jobs is at least $(S_c(k)+1)\eps\pho s_cw_c$ and is about to exceed $(S_c(k)+2)\eps\pho s_cw_c$, for $c$ that is the configuration of the corresponding machine. Let $\mathcal{M}$ be the machine set of this suffix of machines. If the claim does not holds, then $\mathcal{M}\neq \emptyset$.  Denote by $\mathcal{R}$ the sizes of jobs that are small only for machines in $\mathcal{M}$. For the last machine in $M\backslash\mathcal{M}$ (if it exists), the total size of small jobs assigned to the machine, by the algorithm, is at most $(S_c(k)+1)\eps\pho s_cw_c$. This means that the total size of unassigned jobs that were small for this machine was less that $(S_c(k)+1)\eps\pho s_cw_c$. Let the last small job assigned to such a machine be $j'$. Then the last job that was small for this machine was $j'$ and all jobs of size greater than the size of $j'$ are small only for  the machines succeeding this machine in the ordering. Moreover, $j'$ was the last job of size not larger than $p_{j'}$ (along $J_k$) or we have already scheduled the maximum number of jobs of these sizes to this period, else the algorithm would assign more jobs of these sizes to this machine by the choice of the suffix of machines. 

The MILP solution was able to assign $x_{jkc}$ fraction of all jobs of sizes in $\mathcal{R}$ from $J_k$ to all configurations assigned to machines in $\mathcal{M}$ such that for each configuration the total size of small jobs on each configuration is at most $(S_c(k)+1)\eps\pho s_cw_c$ times $z_c$, from constraint \eqref{small_jobs}. Furthermore, the jobs with sizes in $\mathcal{R}$ could not be assigned to any machine in $M
\backslash\mathcal{M}$ as small jobs, since these jobs were big jobs for all those machines. The total (fractions) of configuration counters selected by the MILP solution for this assignment is no more than the number of machines in $\mathcal{M}$, and we get a contradiction to our assumption that $\mathcal{M}\neq \emptyset$ and the claim follows. 
\qed \end{proof}

Observe that at the end of the second rounding step a machine with configuration $c$ has load of at most $(1+2\eps)w_c$.  To see this fact observe that the first rounding step increases the load of a machine with configuration $c$ by no more than $\tilde{\mu} \cdot 2\eps\pho w_c\leq \eps w_c$ by our choice of $\rho=10$, using the definition of $\tilde{\mu} \leq \frac{2}{\eps^2}$.  The second rounding step does not increase the load of machines so this upper bounds on the load holds at the end of the second rounding step.  Since the third rounding step and the final rounding step only move jobs to machine $\mathcal{i}$ the last bound on the load of other machines continue to hold, and we consider the impact of these rounding steps on the load of $\mathcal{i}$.

\begin{lemma}
The total size of jobs scheduled to $\mathcal{i}$ in the third rounding step is at most $\eps \copt$.
\end{lemma}
\begin{proof}
The total size of small jobs left unassigned in a given period is at most
$\sum_{\ell=0}^{\infty} (\oeps)^{-\ell}\eps\pho w s =(\oeps)\eps^{\rho-1}ws$,
where $w$ is the largest load bound among all the configurations where this particular period is active  and $s$ is the largest speed among all configurations with load bound $w$. 
The total size of small jobs left unassigned  for a given load bound $w$ is at most 
$(\oeps)\eps^{\rho-1}ws\cdot \tilde{\mu} \leq (\oeps)\eps^{\rho-1}w\cdot \tilde{\mu}$.
Then, the total size of small jobs left unassigned and thus moved to be processed on machine $\mathcal{i}$ is at most
$\tilde{\mu}\sum_{\ell=0}^{\infty}(\oeps)^{(2-\ell)}\eps^{\rho-1}\copt = (1+\eps)^3\eps^{\rho-2}\tilde{\mu}\copt \leq \eps^{\rho-4}\tilde{\mu}\copt$.
Notice that
$\log_{(\oeps)}\frac{(\oeps)}{\eps} \leq \frac{2}{\eps^2}$.
Since $\rho = 10$, we get that the total size of small jobs left unassigned in all the possible periods is at most
$\eps^{\rho-4}\frac{2}{\eps^2}\copt \leq \eps^{\rho-7}\copt = \eps^3\copt $.

In addition to that some sizes are large (for at least one machine active at the given period) and are perhaps unassigned at the end of the given phase.  For a given pair of $(w,s)$ where $w$ is the largest load bound among all configurations active at a given period, we can assume that $w\leq \eps^{\rho} \copt$.  Using the same sequence of inequalities as the above case we have the following. The total size of jobs left unassigned in a given period is at most
$\sum_{\ell=0}^{\infty} (\oeps)^{-\ell} (1+\eps) w s =\frac{(\oeps)^2}{\eps}ws$,
where $w$ is the largest load bound among all the configurations where this particular period is active  and $s$ is the largest speed among all configurations with load bound $w$. 
The total size of jobs left unassigned  for a given load bound $w$ is at most 
$\frac{(\oeps)^2}{\eps}ws\cdot \tilde{\mu} \leq \frac{(\oeps)^2}{\eps}w\cdot \tilde{\mu}$.
Then, the total size of jobs left unassigned that are large for at least one machine in their period and thus moved to be processed on machine $\mathcal{i}$ is once again at most
$\eps^{\rho} \copt \cdot \tilde{\mu}\cdot \frac{1}{\eps} \cdot \sum_{\ell=0}^{\infty}(\oeps)^{(2-\ell)} = (1+\eps)^3\eps^{\rho-2}\tilde{\mu}\copt \leq \eps^{\rho-4}\tilde{\mu}\copt$.
By $2\eps^3\leq \eps$ the claim follows.
\qed \end{proof}

\paragraph{Bounding the total size of jobs assigned to virtual machines.}
It remains to upper bound the increase of the load of machine $\mathcal{i}$ in the final rounding step.
Let 
\begin{align*}
\kappa = \frac{\eps\sopt}{\left(1 + 2\eps\right)\cdot \left(\frac{1+\eps}{\eps\pho}+1\right)^{\tilde{\mu}(\lambda+1)}\cdot \left(\frac{(\oeps)^3}{\eps^2}\right)}\ ,
\end{align*}
and note that indeed the required property of $\kappa \leq\eps^{2\rho+3}=\eps^{23}$ indeed holds.

\begin{lemma}
The total size of jobs on all virtual machines is at most $2\eps\copt\sopt$.
\end{lemma}
\begin{proof}
Since there is at most one virtual machine assigned each configuration and $\mathcal{i}$ is not virtual, the total size of jobs on all virtual machines with load $w$ and speed $s$ is at most
$\left(1 + 2\eps\right) w \cdot s\cdot \left(\frac{1+\eps}{\eps\pho}+1\right)^{\tilde{\mu}(\lambda+1)}$.  Next we sum these bounds over all $w$, and conclude that the total size of jobs on all virtual machines of speed $s$ where machines of speed $s$ are slow machines is at most \\
$\left(1 + 2\eps\right) s\cdot \left(\frac{1+\eps}{\eps\pho}+1\right)^{\tilde{\mu}(\lambda+1)}\cdot \left(\sum_{l=0}^{\infty}(\oeps)^{1-l}\copt\right)= \left(1 + 2\eps\right) s\cdot \left(\frac{1+\eps}{\eps\pho}+1\right)^{\tilde{\mu}(\lambda+1)}\cdot \left(\frac{(\oeps)^2\copt}{\eps}\right)$.

The total size of jobs on all virtual machines of slow speed is at most
\begin{align*}
&\left(1 + 2\eps\right)\cdot\left(\frac{1+\eps}{\eps\pho}+1\right)^{\tilde{\mu}(\lambda+1)}\cdot \left(\frac{(\oeps)^2\copt}{\eps}\right)\cdot \left(\sum_{l=0}^{\infty}(\oeps)^{-l}\kappa\right)\\
&= \left(1 + 2\eps\right)\cdot\left(\frac{1+\eps}{\eps\pho}+1\right)^{\tilde{\mu}(\lambda+1)}\cdot \left(\frac{(\oeps)^3\kappa\copt}{\eps^2}\right)\ .
\end{align*}

Next we consider the total size of jobs assigned to virtual machines of fast machines.  Then, by the requirement that the MILP solution has integral configuration counters for fast machines and large loads, these configurations of virtual machines that are fast are of small loads.  The total size of jobs on all virtual machines of speed $s$ and small load is at most
$$\left(1 + 2\eps\right) s\cdot \left(\frac{1+\eps}{\eps\pho}+1\right)^{\tilde{\mu}(\lambda+1)}\cdot \left(\sum_{l=0}^{\infty}(\oeps)^{-l}\kappa\copt\right) = \left(1 + 2\eps\right) s\cdot \left(\frac{1+\eps}{\eps\pho}+1\right)^{\tilde{\mu}(\lambda+1)}\cdot \left(\frac{(\oeps)\kappa\copt}{\eps}\right).$$
Therefore, the total load on all virtual machines of fast speed (speed at least $\kappa$ and at most $1$) and small load is at most
$\left(1 + 2\eps\right)\cdot \left(\frac{1+\eps}{\eps\pho}+1\right)^{\tilde{\mu}(\lambda+1)}\cdot \left(\frac{(\oeps)\kappa\copt}{\eps}\right)\cdot
\left(\sum_{l=\lfloor\log_{(\oeps)}\kappa\rfloor}^{0}(\oeps)^{l}\right) \leq \left(1 + 2\eps\right)\cdot \left(\frac{1+\eps}{\eps\pho}+1\right)^{\tilde{\mu}(\lambda+1)}\cdot \left(\frac{(\oeps)\kappa\copt}{\eps}\right)\cdot\left(\sum_{l=0}^{\infty}(\oeps)^{-l}\right)
\leq \left(1 + 2\eps\right)\cdot \left(\frac{1+\eps}{\eps\pho}+1\right)^{\tilde{\mu}(\lambda+1)}\cdot \left(\frac{(\oeps)^3\kappa\copt}{\eps^2}\right)$.
The total load on all virtual machines is at most
\begin{align*}
2\left(1 + 2\eps\right)\cdot \left(\frac{1+\eps}{\eps\pho}+1\right)^{\tilde{\mu}(\lambda+1)}\cdot \left(\frac{(\oeps)^3\kappa\copt}{\eps^2}\right)\ .
\end{align*}
Thus using the definition of $\kappa$, this is at most
$2\eps\sopt\copt$ as we claimed.
\qed \end{proof}

The configuration assigned to machine $\mathcal{i}$ has a load bound of at least $\frac{\copt}{\oeps}$ and at most $\copt \cdot(\oeps)$.  Therefore, after the final rounding step, the load of machine $\mathcal{i}$ is at most $(1+3\eps)\copt+\eps\copt+2\eps\copt = (1+6\eps)\cdot \copt$.  Let $\Lambda_i$ be the load of machine $i$ in the output schedule and note that if $i$ has assigned configuration $c$ then $\Lambda_i \leq (1+2\eps)w_c+3\eps\copt \leq (1+6\eps)w_c$.
Thus the objective value of the output schedule is 
\begin{align*}
\max_{i\in M}\Lambda_i + \sum_{i\in M}\Lambda_i^\phi \leq (1+6\eps)\copt + \sum_{c\in C} z_c\cdot (1+6\eps)^\phi w_c \leq (1+6\eps)^\phi obj(\textsf{sol})
\end{align*} and we conclude that Theorem \ref{theorem_2} holds.

\paragraph{Summary of the scheme.}
The algorithm initially uses the rounding to simplify the instance and uses a guessing step to guess information from the rounded instance to create the configurations and the MILP. Using the generated configurations and the MILP the algorithm finds the least cost solution from among all the MILP solutions for each value of the guessed information. This solution is then transformed into a schedule.
The algorithm for solving the MILP has a complexity of $2^{O(d\log d)}\cdot poly(n)$ using Lenstra's algorithm where $d$ is the number of variables that are required to be integer in the MILP, and
\[
	d = \left(\log_{(\oeps)} \frac{1}{\kappa} + 2\right)^2\cdot\left(\frac{1+\eps}{\eps\pho}+1\right)^{\tilde{\mu}(\lambda+1)}\ = f'\left(\frac{1}{\eps}\right)\ ,
\]
since $\kappa$, $\tilde{\mu}$, and $\lambda$ are functions of $\eps$.  Alll other steps of the algorithm runs in polynomial time and the number of possibilities of the guessed information is also upper bounded by a polynomial of the input encoding length. 
Thus the complexity of the scheme is $f\left(\frac{1}{\eps}\right)\cdot poly(n)$, where $f\left(\frac{1}{\eps}\right)$ is a doubly exponential function in $\frac{1}{\eps}$.

Theorem \ref{theorem_1} guarantees feasibility of the solution of the algorithm and using Theorem \ref{theorem_2} the schedule can be generated for the least cost solution with the proved approximation ratio.  Thus, we have established our result stated as follows.

\begin{theorem}\label{theorem_3}
Problem \prob\ admits an EPTAS.
\end{theorem}

\end{document}